\documentclass[12pt,a4paper]{amsart}
\usepackage[hmargin=2.5cm,vmargin=2.5cm]{geometry}
\usepackage{graphicx} 
\graphicspath{{./Figures}}

\usepackage{amssymb}
\usepackage{amsfonts}
\usepackage{amsmath}
\newtheorem{theorem}{Theorem}

\newtheorem{lemma}{Lemma}

\newtheorem{proposition}{Proposition}
\newtheorem{remark}{Remark}

\usepackage{color}

\begin{document}

\title[Discrete Boltzmann Equation for Anyons]{Discrete Boltzmann Equation for Anyons}

\author[N. Bernhoff]{Niclas Bernhoff}
\address{N.B. Department of Mathematics and Computer Science, Karlstad University, Universitetsgatan 2, 65188 Karlstad, Sweden}
\email{niclas.bernhoff@kau.se}

\keywords{anyons, Haldane statistics, discrete Boltzmann equation, trend to equilibrium}

\maketitle

\textbf{Abstract:} A semi-classical approach to the study of the evolution of anyonic excitations--elementary particles with fractional statistics, complementing bosons and fermions--is through the Boltzmann equation for anyons. This work reviews a discretized version--a system of partial differential equations--of such a quantum equation. Trend to equilibrium is studied for a planar stationary system, as well as the spatially homogeneous system. Essential properties of the linearized operator are proven, implying that results for general steady half-space problems for the discrete Boltzmann equation in a slab geometry can be applied. 

\section{Introduction}

In quantum mechanics the elementary particles, quantum particles, are (at
least, traditionally) either bosons or fermions, if one consider a space of
three (or more) dimensions. Nevertheless, in a space of dimension two (or
one), there are also other possibilities, as was first noted by Leinaas and
Myrheim \cite{LM-77}. Those latter quantum particles, obeying a fractional
statistics, were by Wilczek \cite{Wi-82} named anyons. In 1928 Nordheim
presented the Nordheim-Boltzmann equation \cite{No-28}, a semi-classical
quantum Boltzmann equation for bosons and fermions, also known as the
Uehling-Uhlenbeck equation \cite{UU-33} in literature. In 1995, Bhaduri,
Bhalerao, and Murthy generalized the Nordheim-Boltzmann equation for bosons
and fermions, to yield also for particles obeying Haldane statistics \cite%
{Ha-91}, or fractional exclusion statistics, by a suitable modification \cite%
{BBM-96}. Mathematical studies of this equation has been conducted in, e.g., 
\cite{Ar-13, AN-15, AN-19}. In this paper we review a general discrete model
of Boltzmann equation for anyons--or Haldane statistics--already addressed
in a shorter presentation in \cite{Be-19}.

The equation is introduced in Sect.\ref{S2}. The equilibrium distributions
are characterized in Sect.\ref{S2.2}, and in Sect.\ref{S3} trend to
equilibrium is shown in spatially homogeneous case in Sect.\ref{S3.2} and
planar stationary case in Sect.\ref{S3.1}. The linearized collision operator
is considered in Sect.\ref{S4}, and basic important properties of it are
proven in Sect.\ref{S4.1}.

The main extensions to the previous work \cite{Be-19} are the inclusions of
a proof of Theorem \ref{T7} and of Lemma \ref{L1}--stating the uniqueness of
the equilibrium distribution for given moments--sharpening the statement of
Theorem \ref{T8}--proved analogously to Theorem \ref{T7}.

\section{Discrete\ Boltzmann\ equation\ for\ Haldane\ statistics \label{S2}}

The discrete Boltzmann equation for anyons--or, particles obeying Haldane
statistics--reads \cite{Be-19} 
\begin{equation}
\frac{\partial F_{i}}{\partial t}+\mathbf{p}_{i}\cdot \nabla _{\mathbf{x}%
}F_{i}=Q_{i}^{\alpha }(F)\text{ for }i\in \left\{ 1,...,N\right\} 
\label{ln1}
\end{equation}%
for some real number $\alpha \in \left( 0,1\right) $ and given finite set $%
\mathcal{P}=\left\{ \mathbf{p}_{1},...,\mathbf{p}_{N}\right\} \subset 
\mathbb{R}^{d}$, where $F=(F_{1},...,F_{N})$, with components $%
F_{i}=F_{i}\left( \mathbf{x},t\right) =F\left( \mathbf{x},\mathbf{p}%
_{i},t\right) $ restricted by $0<F_{i}<\dfrac{1}{\alpha }$, is the
distribution function of the particles. For generality, the mathematical
results obtained here are stated for any dimension $d$. The limiting cases $%
\alpha =0$ (without any upper bound on $F_{i}$) corresponding to the
discrete Nordheim-Boltzmann equation for bosons, and $\alpha =1$
corresponding to the discrete Nordheim-Boltzmann equation for fermions can
be included as well, see \cite{Be-17,Be-23g}. Here it is assumed that the
gas is rarefied--imposing only binary interactions between particles to be
considered--and the lack of external forces. Note that vanishing and
saturated states--i.e., $F_{i}=0$ and $F_{i}=\dfrac{1}{\alpha }$ ($\alpha
\neq 0$) for some $i\in \left\{ 1,...,N\right\} $--are excluded due to
technical reasons, that will be addressed further below in Remark \ref{Rem2}-%
\ref{Rem3}.

\begin{remark}
\label{N1}Below we apply the following convention: for a function $g=g(%
\mathbf{p})$ (possibly depending on more variables than $\mathbf{p}$), we
identify $g$ with its restrictions to the points $\mathbf{p}\in \mathcal{P}$%
, i.e.,%
\begin{equation*}
g=\left( g_{1},...,g_{N}\right) ,\text{ where }g_{1}=g\left( \mathbf{p}%
_{1}\right) ,...,g_{N}=g\left( \mathbf{p}_{N}\right) \text{.}
\end{equation*}
\end{remark}

\subsection{Collision operator \label{S2.1}}

The collision operators $Q_{i}^{\alpha }(F)$ are for $i\in \left\{
1,...,N\right\} $ given by 
\begin{equation}
Q_{i}^{\alpha }\left( F\right) =\sum\limits_{j,k,l=1}^{N}\Gamma
_{ij}^{kl}\left( F_{k}F_{l}\Psi _{\alpha }\left( F_{i}\right) \Psi _{\alpha
}(F_{j})-F_{i}F_{j}\Psi _{\alpha }\left( F_{k}\right) \Psi _{\alpha }\left(
F_{l}\right) \right)  \label{l2}
\end{equation}%
where the filling factor $\Psi _{\alpha }\left( y\right) $ is given by 
\begin{equation*}
\Psi _{\alpha }\left( y\right) =\left( 1-\alpha y\right) ^{\alpha }\left(
1+\left( 1-\alpha \right) y\right) ^{1-\alpha }\text{.}
\end{equation*}%
It is assumed that the collision coefficients $\Gamma _{ij}^{kl}$ satisfy
the symmetry relations, due to indistinguishability of particles and
microreversibility, 
\begin{equation}
\Gamma _{ij}^{kl}=\Gamma _{ji}^{kl}=\Gamma _{kl}^{ij}\geq 0
\label{l6}
\end{equation}%
for any indices $\left\{ i,j,k,l\right\} \subset \left\{ 1,...,N\right\} $;
the collision coefficients vanishing, unless the conservation laws%
\begin{equation}
\mathbf{p}_{i}+\mathbf{p}_{j}=\mathbf{p}_{k}+\mathbf{p}_{l}\text{ and }%
\left\vert \mathbf{p}_{i}\right\vert ^{2}+\left\vert \mathbf{p}%
_{j}\right\vert ^{2}=\left\vert \mathbf{p}_{k}\right\vert ^{2}+\left\vert 
\mathbf{p}_{l}\right\vert ^{2}  \label{l3}
\end{equation}%
are satisfied, imposing conservation of momentum and kinetic energy under
interactions of particles. Also mass--or, the number of particles--is
trivially conserved due to form of the collision operator $\left( \ref{l2}%
\right) $. For (the limiting cases) bosons ($\alpha =0$) and fermions ($%
\alpha =1$), the classical filling factors%
\begin{equation*}
\Psi _{0}\left( y\right) =1+y\text{ and }\Psi _{1}\left( y\right) =1-y\text{,%
}
\end{equation*}%
respectively, are recovered, and, e.g., for semions ($\alpha =1/2$) the
filling factor becomes%
\begin{equation*}
\Psi _{1/2}\left( y\right) =\sqrt{1-\frac{y^{2}}{4}}.
\end{equation*}

Denote by $\left\langle \cdot ,\cdot \right\rangle $ the standard scalar
product in $\mathbb{R}^{N}$. Due to symmetry relations $\left( \ref{l6}\right) $, we
have the following proposition for the weak form 
\begin{equation}
\left\langle H,Q^{\alpha }\left( F\right) \right\rangle
=\sum\limits_{i,j,k,l=1}^{N}\Gamma _{ij}^{kl}H_{i}\left( F_{k}F_{l}\Psi
_{\alpha }\left( F_{i}\right) \Psi _{\alpha }(F_{j})-F_{i}F_{j}\Psi _{\alpha
}\left( F_{k}\right) \Psi _{\alpha }\left( F_{l}\right) \right)  \label{l4b}
\end{equation}%
of the collision operator.

\begin{proposition}
For any function $H=H(\mathbf{p})$ expression $\left( \ref{l4b}\right) $ can
be recast as%
\begin{multline}
\left\langle H,Q^{\alpha }\left( F\right) \right\rangle =\frac{1}{4}%
\sum\limits_{i,j,k,l=1}^{N}\Gamma _{ij}^{kl}\left(
H_{i}+H_{j}-H_{k}-H_{l}\right)  \\
\times \left( F_{k}F_{l}\Psi _{\alpha }\left( F_{i}\right) \Psi _{\alpha
}(F_{j})-F_{i}F_{j}\Psi _{\alpha }\left( F_{k}\right) \Psi _{\alpha }\left(
F_{l}\right) \right) .  \label{l4c}
\end{multline}
\end{proposition}

\subsection{Collision invariants and equilibrium distributions \label{S2.2}}

A collision invariant is a function $\phi =\phi \left( \mathbf{p}\right) $,
such that%
\begin{equation}
\phi _{i}+\phi _{j}=\phi _{k}+\phi _{l}  \label{c9c}
\end{equation}%
for all indices $\left\{ i,j,k,l\right\} \subset \left\{ 1,...,N\right\} $
such that $\Gamma _{ij}^{kl}\neq 0$. Trivially--by conservation of mass,
momentum, and kinetic energy--the set of collision invariants include all
functions of the form 
\begin{equation}
\phi =a+\mathbf{b\cdot p}+c\left\vert \mathbf{p}\right\vert ^{2}  \label{l5}
\end{equation}%
for some $a,c\in \mathbb{R}$ and $\mathbf{b}\in \mathbb{R}^{d}$. Note that
by Remark $\ref{N1}$ and in correspondence with relations $\left( \ref{c9c}%
\right) $ the collision invariants $\phi =\phi (\mathbf{p})$ given in $%
\left( \ref{l5}\right) $ are vectors.

In general, in the discrete case, there can be also so called spurious--or,
"non-physical"--collision invariants. This is a common problem for different
kinds of velocity/momentum models, cf. \cite{BV-08}; if there are not enough
of admissible collisions, unwanted quantities $\phi =\phi \left( \mathbf{p}%
\right) $ will be invariant under interactions-- the most trivial case: with
no admissible collisions at all, all functions $\phi =\phi \left( \mathbf{p}%
\right) $ will be collision invariants. In fact, to obtain only the desired
set of collision invariants, there must be a set of $N-p$--here $p$ denotes
the number of desired collision invariants--independent admissible
collisions, i.e., collisions with non-zero collision coefficients, that can
not be obtained by any chain of other collisions in the set (or their
reversion). Discrete models without spurious collision invariants are called
normal and methods of their construction have been extensively studied, see
for example \cite{BC-99,BV-08,BV-16} and references therein. Consider
below--even if this restriction is not necessary in the general
context--only normal discrete models. That is, consider discrete models
without spurious collision invariants, i.e., any collision invariant is of
the form $\left( \ref{l5}\right) $. For normal discrete models the equation%
\begin{equation}
\left\langle Q^{\alpha }\left( F\right) ,\phi \right\rangle =0  \label{c4a}
\end{equation}%
has the general solution $\left( \ref{l5}\right) $.

With $H=\log \dfrac{F}{\Psi _{\alpha }\left( F\right) }$ in expression $%
\left( \ref{l4c}\right) $, we can recast the expression to%
\begin{multline}
\left\langle \log \frac{F}{\Psi _{\alpha }\left( F\right) },Q^{\alpha
}\left( F\right) \right\rangle =\frac{1}{4}\sum\limits_{i,j,k=1}^{N}\Gamma
_{ij}^{kl}\Psi _{\alpha }\left( F_{i}\right) \Psi _{\alpha }(F_{j})\Psi
_{\alpha }\left( F_{k}\right) \Psi _{\alpha }\left( F_{l}\right)  \\
\times \log \frac{F_{i}F_{j}\Psi _{\alpha }\left( F_{k}\right) \Psi _{\alpha
}\left( F_{l}\right) }{F_{k}F_{l}\Psi _{\alpha }\left( F_{i}\right) \Psi
_{\alpha }(F_{j})}\left( \frac{F_{k}}{\Psi _{\alpha }\left( F_{k}\right) }%
\frac{F_{l}}{\Psi _{\alpha }\left( F_{l}\right) }-\frac{F_{i}}{\Psi _{\alpha
}\left( F_{i}\right) }\frac{F_{j}}{\Psi _{\alpha }(F_{j})}\right) \leq 0%
\text{.}  \label{c17e}
\end{multline}%
Inequality $\left( \ref{c17e}\right) $ is obtained by the relation 
\begin{equation}
\left( z-y\right) \log \frac{y}{z}\leq 0  \label{c6}
\end{equation}%
for all positive numbers $y$ and $z$, where it is actual equality if and
only if $y=z$. It follows that there is equality in inequality $\left( \ref%
{c17e}\right) $ if and only if%
\begin{equation}
\frac{F_{i}}{\Psi _{\alpha }\left( F_{i}\right) }\frac{F_{j}}{\Psi _{\alpha
}(F_{j})}=\frac{F_{k}}{\Psi _{\alpha }\left( F_{k}\right) }\frac{F_{l}}{\Psi
_{\alpha }\left( F_{l}\right) }  \label{c17f}
\end{equation}%
for all indices $\left\{ i,j,k,l\right\} \subset \left\{ 1,...,N\right\} $
such that $\Gamma _{ij}^{kl}\neq 0$.

A Maxwellian distribution--or, Maxwellian--is a function $M=M(\mathbf{p})$
of the form 
\begin{equation*}
M=e^{-\phi }=Ke^{-\mathbf{b\cdot p}-c\left\vert \mathbf{p}\right\vert ^{2}}%
\text{, with }K=e^{-a}>0\text{, }
\end{equation*}%
or, equivalently, 
\begin{equation*}
M_{i}=e^{-\phi _{i}}=Ke^{-\mathbf{b\cdot p}_{i}-c\left\vert \mathbf{p}%
_{i}\right\vert ^{2}}\ \text{for }i\in \left\{ 1,...,N\right\} \text{,}
\end{equation*}%
where $\phi =\left( \phi _{1},...,\phi _{N}\right) $ is a collision
invariant $\left( \ref{l5}\right) $. There is equality in inequality $\left( %
\ref{c17e}\right) $ if and only if $\log \dfrac{F}{\Psi _{\alpha }\left(
F\right) }$ is a collision invariant--noted by taking the logarithm of
equality $\left( \ref{c17f}\right) $--or, equivalently, if and only if $%
\dfrac{F}{\Psi _{\alpha }\left( F\right) }$ is a Maxwellian $M$, i.e., the
equilibrium distributions $P$ are given by the transcendental equation, see 
\cite{Wu-94} for the continuous case, 
\begin{equation}
\dfrac{P}{\Psi _{\alpha }\left( P\right) }=M\text{.}  \label{c14b}
\end{equation}

\begin{proposition}
\label{Prop2}The equilibrium distributions of system $\left( \ref{ln1}%
\right) $ are characterized by system $\left( \ref{c14b}\right) $.
\end{proposition}

Note that, by solving system $\left( \ref{c14b}\right) $, for bosons ($%
\alpha =0$) and fermions ($\alpha =1$), it is found that the equilibrium
distributions are the Planckians 
\begin{equation*}
P=\frac{M}{1-M}\text{ and }P=\frac{M}{1+M}\text{,}
\end{equation*}%
respectively. Moreover, for semions ($\alpha =1/2$) it renders in the
equilibrium distribution%
\begin{equation*}
P=\frac{2M}{\sqrt{4+M^{2}}}\text{.}
\end{equation*}

\begin{remark}
\label{Rem2}The exclusion of vanishing and saturated states is required for
the categorization of equilibrium distributions in Proposition $\ref{Prop2}$%
. In the continuous case, saturated states--in form of step distributions,
being vanishing, i.e., equal to zero, above and saturated, i.e., equal to $%
\dfrac{1}{\alpha }$, below a microscopic energy treshold--might as in the
case of fermions appear for anyons as well \cite{Wu-94}. However, in the
discrete case, the general categorization of the equilibrium distributions
is not straightforward due to the finiteness of allowed collisions. In fact,
allowing vanishing states, would not only add the trivial equilibrium
distribution $M=0$, but also--unlike in the continuous case--other
non-trivial equilibrium distributions. For example, letting a chosen set of
components to be zero, may even a few components of the equilibrium
distribution to be chosen arbitrarily (below the upper bound for $\alpha
\neq 0$). In fact, any pair of components $F_{i}$ and $F_{j}$, such that $%
\Gamma _{ij}^{kl}=0$ for all $\left\{ k,l\right\} \in \left\{
1,...,N\right\} \diagdown \left\{ i,j\right\} $ (such pairs exist for any
finite set $\mathcal{P}$), can be chosen arbitrarily (possibly also more
components), while the rest being zero. 

The general equilibrium distributions (with vanishing and saturated states)
are given by   
\begin{equation}
F_{k}F_{l}\Psi _{\alpha }\left( F_{i}\right) \Psi _{\alpha
}(F_{j})=F_{i}F_{j}\Psi _{\alpha }\left( F_{k}\right) \Psi _{\alpha }\left(
F_{l}\right)   \label{c17g}
\end{equation}%
for all indices $\left\{ i,j,k,l\right\} \subset \left\{ 1,...,N\right\} $
such that $\Gamma _{ij}^{kl}\neq 0$. Here a particular relation $\left( \ref%
{c17g}\right) $ may either be of the form zero equals zero, or, otherwise it
is equivalent to the corresponding relation $\left( \ref{c17f}\right) $.
There might very well be combinations of those two alternatives for a
general equilibrium distribution, complicating a general classification of
equilibrium distributions. 
\end{remark}

\section{$\mathcal{H}$-functional(s)\ and\ trend\ to\ equilibrium \label{S3}}

This section concerns the trend to equilibrium in two particular cases: the
planar stationary case and the spatially homogeneous case.

\subsection{Planar stationary system \label{S3.1}}

Introduce a modified $\mathcal{H}$-functional%
\begin{equation*}
\widetilde{\mathcal{H}}[F]=\widetilde{\mathcal{H}}[F](x)=\sum%
\limits_{i=1}^{N}p_{i}^{1}\mu (F_{i}(x)),
\end{equation*}%
where, cf. \cite{AN-15}, 
\begin{equation}
\mu (y)=y\log y+\left( 1-\alpha y\right) \log \left( 1-\alpha y\right)
-\left( 1+\left( 1-\alpha \right) y\right) \log \left( 1+\left( 1-\alpha
\right) y\right) .  \label{c15}
\end{equation}%
Note that%
\begin{equation}
\mu ^{\prime }(y)=\log \frac{y}{\Psi _{\alpha }\left( y\right) }\text{ and }%
\mu ^{\prime \prime }(y)=\frac{1}{y\left( 1-\alpha y\right) \left( 1+\left(
1-\alpha \right) y\right) }>0\text{.}  \label{d1}
\end{equation}

Any solution to the planar stationary system%
\begin{equation}
B\dfrac{dF}{dx}=Q^{\alpha }\left( F\right) \text{, where }B=\text{\textrm{%
diag}}(p_{1}^{1},...,p_{N}^{1})\text{, with }x\in \mathbb{R}_{+}\text{,}
\label{e1}
\end{equation}%
satisfies the inequality 
\begin{equation}
\frac{d}{dx}\widetilde{\mathcal{H}}[F]=\sum\limits_{i=1}^{N}p_{i}^{1}\frac{%
dF_{i}}{dx}\log \frac{F_{i}}{\Psi _{\alpha }\left( F_{i}\right) }%
=\left\langle \log \frac{F}{\Psi _{\alpha }\left( F\right) },Q^{\alpha
}\left( F\right) \right\rangle \leq 0\text{,}  \label{h1}
\end{equation}%
with equality in inequality $\left( \ref{h1}\right) $ if and only if $F$ is
an equilibrium distribution $\left( \ref{c14b}\right) $. Introduce the fluxes%
\begin{equation}
\left\{ 
\begin{array}{l}
\widetilde{j}_{1}=\left\langle B\mathbf{1},F\right\rangle \\ 
\widetilde{j}_{i+1}=\left\langle Bp^{i},F\right\rangle \text{ for }i\in
\left\{ 1,...,d\right\} \\ 
\widetilde{j}_{d+2}=\left\langle B\left\vert \mathbf{p}\right\vert
^{2},F\right\rangle%
\end{array}%
.\right.  \label{c3}
\end{equation}%
Applying relation $\left( \ref{c4a}\right) $ to system $\left( \ref{e1}%
\right) $, implies that the fluxes $\widetilde{j}_{1},...,\widetilde{j}%
_{d+2} $ are independent of $x$ in the planar stationary case. For fixed
numbers $\widetilde{j}_{1},...,\widetilde{j}_{d+2}$ denote by $\mathbb{P}$
the manifold of all equilibrium distributions $F=P$--given by equation $%
\left( \ref{c14b}\right) $--with fluxes $\left( \ref{c3}\right) $. The
following theorem can be proved by arguments similar to the ones used for
the discrete Boltzmann equation in \cite{CIPS-88,BB-07}.

\begin{theorem}
\label{T7}Let $F=F(x)$ be a bounded solution to system $\left( \ref{e1}%
\right) $, and assume that there exists a number $\eta >0$, such that $\eta
\leq F_{i}(x)\leq \dfrac{1}{\alpha }-\eta $ for all $i\in \left\{
1,...,N\right\} $. Then 
\begin{equation*}
\underset{x\rightarrow \infty }{\lim }\mathrm{dist}(F(x),\mathbb{P})=0,
\end{equation*}%
where $\mathbb{P}$ is the manifold of equilibrium distributions with the
same fluxes $\left( \ref{c3}\right) $ as $F$. If there are only finitely
many equilibrium distributions in $\mathbb{P}$, then there is an equilibrium
distribution $P$ in $\mathbb{P}$, such that $\underset{x\rightarrow \infty }{%
\lim }F(x)=P$.
\end{theorem}

\begin{proof}
(cf. \cite{CIPS-88,BB-07}) \textit{The function }$F$\textit{\ is bounded,
and so the derivative }$\dfrac{dF}{dx}$\textit{\ is bounded.} Moreover, the
functional $\widetilde{\mathcal{H}}[F]=\widetilde{\mathcal{H}}[F](x)$ is
bounded--$\mu (y)$ is continuous, non-positive, and bounded below, since $%
\underset{y\rightarrow 0^{+}}{\lim }\mu (y)=0=\mu (0)$ as well as $\underset{%
y\rightarrow \left( 1/\alpha \right) ^{-}}{\lim }\mu (y)=0=\mu (\dfrac{1}{%
\alpha })$, while $\mu ^{\prime }(y)$ is strictly increasing for $0<y<\dfrac{%
1}{\alpha }$--and differentiable in $\mathbb{R}_{+}$. Hence, the limit $%
\underset{x\rightarrow \infty }{\lim }\widetilde{\mathcal{H}}[F]$ exists and
is finite. Consequently, 
\begin{equation*}
\int\limits_{0}^{\infty }\frac{d}{dx}\widetilde{\mathcal{H}}[F]~dx=\underset{%
x\rightarrow \infty }{\lim }\widetilde{\mathcal{H}}[F](x)-\widetilde{%
\mathcal{H}}[F](0)\leq 0
\end{equation*}%
is a finite non-positive number. We want to show that%
\begin{equation*}
\mathrm{dist}(F\left( x_{s}\right) ,\mathbb{P})\rightarrow 0\text{ as }%
s\rightarrow \infty
\end{equation*}%
for any increasing sequence $\left\{ x_{s}\right\} _{s=1}^{\infty }$ of
positive real numbers, such that $x_{s}\rightarrow \infty $ as $s\rightarrow
\infty $. We assume that the assertion is false. Then there are positive
numbers $\epsilon _{1}>0$ and $\delta _{1}>0$, and an increasing sequence $%
\left\{ y_{s}\right\} _{s=1}^{\infty }$ of positive real numbers, such that $%
y_{s+1}-y_{s}\geq \epsilon _{1}$ and \textrm{dist}$(F\left( y_{s}\right) ,%
\mathbb{P})\geq \delta _{1}$. The derivative of $F$ is bounded in $\mathbb{R}%
_{+}$, and therefore, there is a positive number $\epsilon _{2}>0$, such
that $\epsilon _{2}<\dfrac{\epsilon _{1}}{2}$ and \textrm{dist}$(F\left(
x\right) ,\mathbb{P})\geq \dfrac{\delta _{1}}{2}$ if $x\in
I_{s}=[y_{s}-\epsilon _{2},y_{s}+\epsilon _{2}]$ for some $s\in \left\{
1,2,...\right\} $.

We denote 
\begin{equation*}
\Lambda (I_{s})=-\int\limits_{y_{s}-\epsilon _{2}}^{y_{s}+\epsilon _{2}}%
\frac{d}{dx}\widetilde{\mathcal{H}}[F](x)~dx=\widetilde{\mathcal{H}}%
[F](y_{s}-\epsilon _{2})-\widetilde{\mathcal{H}}[F](y_{s}+\epsilon _{2})\geq
0\text{\ for }s\in \left\{ 1,2,...\right\} .
\end{equation*}%
The positive series $\sum\limits_{s=1}^{\infty }\Lambda (I_{s})$ is bounded,
since%
\begin{equation*}
\sum\limits_{s=1}^{\infty }\Lambda (I_{s})\leq -\int\limits_{0}^{\infty }%
\frac{d}{dx}\widetilde{\mathcal{H}}[F]~dx\text{,}
\end{equation*}%
and, hence, the series converges. Therefore, $\Lambda (I_{s})\rightarrow 0$
as $s\rightarrow \infty $, and there must be numbers $z_{s}\in I_{s}$ for $%
s\in \left\{ 1,2,...\right\} $, such that $\dfrac{d}{dx}\widetilde{\mathcal{H%
}}[F](z_{s})\rightarrow 0$ as $s\rightarrow \infty $. The sequence $\left\{
F(z_{s})\right\} _{s=1}^{\infty }$ is bounded, and, hence, by the
Bolzano-Weierstrass theorem, we can extract a subsequence $\left\{
F(z_{s_{r}})\right\} _{r=1}^{\infty }$ such that $\underset{r\rightarrow
\infty }{\lim }F(z_{s_{r}})=G$ exists. Then 
\begin{equation*}
\left\langle \log \frac{G}{\Psi _{\alpha }\left( G\right) },Q\left(
G,G\right) \right\rangle =\,\underset{r\rightarrow \infty }{\lim }%
\left\langle \log \frac{F(z_{s_{r}})}{\Psi _{\alpha }\left(
F(z_{s_{r}})\right) },Q\left( F,F\right) (z_{s_{r}})\right\rangle =\,%
\underset{r\rightarrow \infty }{\lim }\dfrac{d}{dx}\widetilde{\mathcal{H}}%
[F](z_{s_{r}})=0\text{,}
\end{equation*}%
and, hence, $G$ must be an equilibrium distribution $\left( \ref{c14b}%
\right) $. Clearly, $G$ has the same invariant fluxes $\left( \ref{c3}%
\right) $ as $F$, and therefore belongs to $\mathbb{P}$. This is a
contradiction, since \textrm{dist}$(F\left( z_{s_{r}}\right) ,\mathbb{P}%
)\geq \dfrac{\delta _{1}}{2}$ for all $r\in \left\{ 1,2,...\right\} $.
Hence, 
\begin{equation}
\mathrm{dist}(F\left( x\right) ,\mathbb{P})\rightarrow 0\text{ as }%
x\rightarrow \infty \text{.}  \label{sh6}
\end{equation}%
If there are only finitely many equilibrium distributions in $\mathbb{P}$,
then the only possibility for the limit $\left( \ref{sh6}\right) $ to be
satisfied is that $F$ converges to some equilibrium distribution $P$ in $%
\mathbb{P}$.
\end{proof}

\begin{remark}
\label{Rem3} The role of $\eta >0$ above in Theorem $\ref{T7}$ (and below in
Theorem $\ref{T8}$) is that any (sub-)limit distribution, as well as the
filling factor for it, will have non-vanishing components. Existence of such 
$\eta >0$ is an assumption, and will not be possible to prove in general.
However, formally, the domain of the function $\varphi (y)=\log \dfrac{y}{%
\Psi _{\alpha }\left( y\right) }$ in the proof of Theorem $\ref{T7}$ could
be extended to the interval $[0,\dfrac{1}{\alpha }]$ by defining $\varphi
(0)=-\infty $ and $\varphi (\dfrac{1}{\alpha })=\infty $. Then the limit
distribution $G$ will be a general equilibrium distribution, see Remark $\ref%
{Rem2}$, and not necessarily of the form $\left( \ref{c14b}\right) $.
\end{remark}

\subsection{Spatially homogeneous system \label{S3.2}}

For the spatially homogeneous system%
\begin{equation}
\dfrac{dF}{dt}=Q^{\alpha }\left( F\right) \text{, }t\in \mathbb{R}_{+}\text{,%
}  \label{e2}
\end{equation}%
similar results, presented in Theorem $\ref{T8}$ below, for the trend to
equilibrium, can be obtained analogously, now considering instead the $%
\mathcal{H}$-functional%
\begin{equation*}
\mathcal{H}[F]=\mathcal{H}[F](t)=\sum\limits_{i=1}^{N}\mu (F_{i}(t))\text{,}
\end{equation*}%
with $\mu $ still given by expression $\left( \ref{c15}\right) $, and the
moments%
\begin{equation}
\left\{ 
\begin{array}{l}
j_{1}=\left\langle 1,F\right\rangle \\ 
j_{i+1}=\left\langle p^{i},F\right\rangle \text{ for }i\in \left\{
1,...,d\right\} \\ 
j_{d+2}=\left\langle \left\vert \mathbf{p}\right\vert ^{2},F\right\rangle%
\end{array}%
.\right.  \label{c3a}
\end{equation}
Any solution to the spatially homogeneous system $\left( \ref{e2}\right) $
now satisfies the inequality 
\begin{equation*}
\frac{d}{dt}\mathcal{H}[F]=\sum\limits_{i=1}^{N}\frac{dF_{i}}{dt}\log \frac{%
F_{i}}{\Psi _{\alpha }\left( F_{i}\right) }=\left\langle \log \frac{F}{\Psi
_{\alpha }\left( F\right) },Q^{\alpha }\left( F\right) \right\rangle \leq 0%
\text{.}
\end{equation*}

The following result is relevant in the spatially homogeneous case.

\begin{lemma}
\label{L1}Let $P$ and $\widetilde{P}$ be two equilibrium distributions with
the same moments $\left( \ref{c3a}\right) $. Then $P=\widetilde{P}$.
\end{lemma}

\begin{proof}
Let $I=\left\{ 0,...,d+1\right\} $ and $\phi ^{0}=1,\phi ^{1}=p^{1},...,\phi
^{d}=p^{d},\phi ^{d+1}=\left\vert \mathbf{p}\right\vert ^{2}$. Then 
\begin{eqnarray*}
\log \left( \left( P^{-1}-\alpha \right) ^{-\alpha }\left( P^{-1}+1-\alpha
\right) ^{\alpha -1}\right) &=&\log \dfrac{P}{\Psi _{\alpha }\left( P\right) 
}=\sum\limits_{i\in I}c_{i}\phi ^{i}\text{ and} \\
\log \left( \left( \widetilde{P}^{-1}-\alpha \right) ^{-\alpha }\left( 
\widetilde{P}^{-1}+1-\alpha \right) ^{\alpha -1}\right) &=&\log \dfrac{%
\widetilde{P}}{\Psi _{\alpha }\left( \widetilde{P}\right) }%
=\sum\limits_{i\in I}\widetilde{c}_{i}\phi ^{i},
\end{eqnarray*}%
for some numbers $c_{1},...,c_{d+1}$ and $\widetilde{c}_{1},...,\widetilde{c}%
_{d+1}$, while 
\begin{equation*}
\left\langle \phi ^{i},P\right\rangle =j_{i}=\left\langle \phi ^{i},%
\widetilde{P}\right\rangle \text{ for }\,i\in I.
\end{equation*}%
Obviously,%
\begin{eqnarray*}
\left\langle \log \dfrac{P}{\Psi _{\alpha }\left( P\right) },P\right\rangle
&=&-\sum\limits_{i\in I}c_{i}j_{i}=\left\langle \log \dfrac{P}{\Psi _{\alpha
}\left( P\right) },\widetilde{P}\right\rangle \text{ and} \\
\left\langle \log \dfrac{\widetilde{P}}{\Psi _{\alpha }\left( \widetilde{P}%
\right) },P\right\rangle &=&-\sum\limits_{i\in I}\widetilde{c}%
_{i}j_{i}=\left\langle \log \dfrac{\widetilde{P}}{\Psi _{\alpha }\left( 
\widetilde{P}\right) },\widetilde{P}\right\rangle \text{.}
\end{eqnarray*}%
Hence, 
\begin{multline}
\sum\limits_{i=1}^{N}P_{i}\widetilde{P}_{i}\left( \widetilde{P}%
_{i}^{-1}-P_{i}^{-1}\right) \log \left( \left( \frac{P_{i}^{-1}-\alpha }{%
\widetilde{P}_{i}^{-1}-\alpha }\right) ^{-\alpha }\left( \frac{%
P_{i}^{-1}+1-\alpha }{\widetilde{P}_{i}^{-1}+1-\alpha }\right) ^{\alpha
-1}\right) \\
=\sum\limits_{i=1}^{N}\left( P_{i}-\widetilde{P}_{i}\right) \log \frac{%
\left( P_{i}^{-1}-\alpha \right) ^{-\alpha }\left( P_{i}^{-1}+1-\alpha
\right) ^{\alpha -1}}{\left( \widetilde{P}_{i}^{-1}-\alpha \right) ^{-\alpha
}\left( \widetilde{P}_{i}^{-1}+1-\alpha \right) ^{\alpha -1}} \\
=\left\langle \log \dfrac{P}{\Psi _{\alpha }\left( P\right) }-\log \dfrac{%
\widetilde{P}}{\Psi _{\alpha }\left( \widetilde{P}\right) },P-\widetilde{P}%
\right\rangle =0\text{.}  \label{h2}
\end{multline}%
By relation $\left( \ref{c6}\right) $, it follows that 
\begin{multline}
\left( \widetilde{P}_{i}^{-1}-P_{i}^{-1}\right) \log \left( \left( \frac{%
P_{i}^{-1}-\alpha }{\widetilde{P}_{i}^{-1}-\alpha }\right) ^{-\alpha }\left( 
\frac{P_{i}^{-1}+1-\alpha }{\widetilde{P}_{i}^{-1}+1-\alpha }\right)
^{\alpha -1}\right) \\
=\alpha \left( \widetilde{P}_{i}^{-1}-\alpha -\left( P_{i}^{-1}-\alpha
\right) \right) \log \left( \frac{\widetilde{P}_{i}^{-1}-\alpha }{%
P_{i}^{-1}-\alpha }\right) \\
\left( 1-\alpha \right) \left( \widetilde{P}_{i}^{-1}+1-\alpha -\left(
P_{i}^{-1}+1-\alpha \right) \right) \log \left( \frac{\widetilde{P}%
_{i}^{-1}+1-\alpha }{P_{i}^{-1}+1-\alpha }\right) \geq 0\text{ for }i\in I.
\label{ie1}
\end{multline}%
Hence, by equality $\left( \ref{h2}\right) $, it follows that\ $P=\widetilde{%
P}$. Indeed, all the inequalities in $\left( \ref{ie1}\right) $, must be
equalities, and then 
\begin{equation*}
\left( \widetilde{P}_{i}^{-1}-P_{i}^{-1}\right) \log \left( \frac{\widetilde{%
P}_{i}^{-1}-\alpha }{P_{i}^{-1}-\alpha }\right) =\left( \widetilde{P}%
_{i}^{-1}-P_{i}^{-1}\right) \log \left( \frac{\widetilde{P}%
_{i}^{-1}+1-\alpha }{P_{i}^{-1}+1-\alpha }\right) =0\text{,}
\end{equation*}%
implying that $\widetilde{P}_{i}^{-1}=P_{i}^{-1}$ for all $i\in I.$
\end{proof}

\begin{theorem}
\label{T8}Let $F=F(t)$ be a bounded solution to the system $\left( \ref{e2}%
\right) $, and assume that there exist numbers $\eta >0$ and $t_{0}\geq 0$,
such that $\eta \leq F_{i}(t)\leq \dfrac{1}{\alpha }-\eta $ for all $i\in
\left\{ 1,...,N\right\} $ and $t\geq t_{0}$. Then 
\begin{equation*}
\underset{t\rightarrow \infty }{\lim }F(t)=P\text{,}
\end{equation*}%
where $P$ is the equilibrium distribution with the same moments $\left( \ref%
{c3a}\right) $ as $F$.
\end{theorem}

\begin{remark}
\label{R1}Let $I_{N}=\left\{ 1,...,N\right\} $ and $1\leq m\leq n\leq N-m$,
and denote%
\begin{eqnarray}
Q_{i}^{\alpha }\left( F\right) &=&\sum\limits_{1\leq m\leq n\leq
N-m}a_{mn}Q_{i}^{\alpha ,mn}\left( F\right) \text{, with }a_{mn}\geq 0\text{%
, where}  \notag \\
Q_{i}^{\alpha ,mn}\left( F\right) &=&\sum\limits_{\substack{ I^{\prime
},I^{\prime \prime }\subset I_{N}  \\ \left\vert I^{\prime }\right\vert =n%
\text{, }\left\vert I^{\prime \prime }\right\vert =m}}\Gamma _{I^{\prime
}}^{I^{\prime \prime }}\left( \sum\limits_{k\in I^{\prime }}\delta
_{ik}-\sum\limits_{k\in I^{\prime \prime }}\delta _{ik}\right)  \notag \\
&&\times \left( \prod_{j\in I^{\prime }}F_{j}\prod\limits_{j\in I^{\prime
\prime }}\Psi _{\alpha }\left( F_{j}\right) -\prod\limits_{j\in I^{\prime
\prime }}F_{j}\prod_{j\in I^{\prime }}\Psi _{\alpha }\left( F_{j}\right)
\right) \\
&=&\sum\limits_{\substack{ I^{\prime },I^{\prime \prime }\subset I  \\ %
\left\vert I^{\prime }\right\vert =n\text{, }\left\vert I^{\prime
}\right\vert =m}}\Gamma _{I^{\prime }}^{I^{\prime \prime }}\left(
\sum\limits_{k\in I^{\prime }}\delta _{ik}-\sum\limits_{k\in I^{\prime
\prime }}\delta _{ik}\right)  \notag \\
&&\times \prod_{j\in I^{\prime }\cup I^{\prime \prime }}\Psi _{\alpha
}\left( F_{j}\right) \left( \prod_{j\in I^{\prime }}\frac{F_{j}}{\Psi
_{\alpha }\left( F_{j}\right) }-\prod\limits_{j\in I^{\prime \prime }}\frac{%
F_{j}}{\Psi _{\alpha }\left( F_{j}\right) }\right) \text{, } \\
&&\text{with }\Psi _{\alpha }\left( y\right) =\left( 1-\alpha y\right)
^{\alpha }\left( 1+\left( 1-\alpha \right) y\right) ^{1-\alpha }\text{.}
\label{q1}
\end{eqnarray}%
Here $\Gamma _{I^{\prime }}^{I^{\prime \prime }}=0$ if the relations 
\begin{equation*}
\sum\limits_{k\in I^{\prime }}\mathbf{p}_{k}=\sum\limits_{k\in I^{\prime
\prime }}\mathbf{p}_{k}\text{ and }\sum\limits_{k\in I^{\prime }}\left\vert 
\mathbf{p}_{k}\right\vert ^{2}=\sum\limits_{k\in I^{\prime \prime
}}\left\vert \mathbf{p}_{k}\right\vert ^{2}
\end{equation*}%
are not satisfied (can be replaced by other collision invariants as well).
Then, in a similar way as above, we can obtain corresponding results for the
system $\left( \ref{ln1}\right) $, and its restrictions to systems $\left( %
\ref{e1}\right) $ and $\left( \ref{e2}\right) $. In particular, the
stationary points of the systems are still characterized by equation $\left( %
\ref{c14b}\right) $ and (at least versions of) Theorems $\ref{T7}$ and $\ref%
{T8}$ are still valid. Indeed, if at least one $a_{mn}$ such that $m\neq n$
is nonzero, then the collision invariants (for normal models) will be of the
form 
\begin{equation*}
\phi =\mathbf{b\cdot p}+c\left\vert \mathbf{p}\right\vert ^{2},
\end{equation*}%
and we will have to exclude the invariants $\widetilde{j}_{1}$ and $j_{1}$
from the invariants $\left( \ref{c3}\right) $ and $\left( \ref{c3a}\right) $%
, respectively, for Theorem $\ref{T7}$ and Theorem $\ref{T8}$ to stay valid,
cf. \cite{Be-23g}. A drawback is that, in general, it will not be clear how
to construct the sets $\mathcal{P}$ to obtain normal discrete models. An
example when such generalizations (with $\alpha =0$) are of interest is for
excitations in a Bose gas interacting with a Bose-Einstein condensate \cite%
{KD-84,ZNG-99,AN-12,GR-13,AN-13,Be-15, Be-23g}.\ However, even if the
momentum is still assumed to be conserved during a collision, the energy
conserved will (in the general case) be different from the kinetic one
conserved by relations $\left( \ref{l3}\right) $. Furthermore, the equation
will (in the general case) also be coupled by a Gross-Pitaevskii equation
for the density of the condensate \cite{KD-84,ZNG-99,AN-12}.
\end{remark}

\section{Linearized collision operator \label{S4}}

For any $\alpha \in \left[ 0,1\right] $%
\begin{eqnarray*}
&&\Psi _{\alpha }^{\prime }\left( y\right) \\
&=&-\alpha ^{2}\left( 1-\alpha y\right) ^{\alpha -1}\left( 1+\left( 1-\alpha
\right) y\right) ^{1-\alpha }+\left( 1-\alpha \right) ^{2}\left( 1-\alpha
y\right) ^{\alpha }\left( 1+\left( 1-\alpha \right) y\right) ^{-\alpha } \\
&=&\Psi _{\alpha }\left( y\right) \left( \frac{1-2\alpha -\alpha \left(
1-\alpha \right) y}{\left( 1-\alpha y\right) \left( 1+\left( 1-\alpha
\right) y\right) }\right) =\Psi _{\alpha }\left( y\right) \left( \frac{%
1-2\alpha -\alpha \left( 1-\alpha \right) y}{1+y\left( 1-2\alpha -\alpha
\left( 1-\alpha \right) y\right) }\right) \text{,}
\end{eqnarray*}%
and, hence, 
\begin{equation}
\frac{\Psi _{\alpha }\left( P_{i}\right) -\Psi _{\alpha }^{\prime }\left(
P_{i}\right) P_{i}}{P_{i}\Psi _{\alpha }\left( P_{i}\right) }=\frac{1}{%
P_{i}\left( 1-\alpha P_{i}\right) \left( 1+\left( 1-\alpha \right)
P_{i}\right) }\text{ for }i\in \left\{ 1,...,N\right\} \text{.}  \label{l20}
\end{equation}%
Furthermore, substituting%
\begin{equation}
F=P+R^{1/2}f\text{, with }R=P\left( 1-\alpha P\right) \left( 1+\left(
1-\alpha \right) P\right) \text{ and }\dfrac{P}{\Psi _{\alpha }\left(
P\right) }=M\text{,}  \label{l20a}
\end{equation}%
in system $\left( \ref{ln1}\right) $, and ignoring all terms of second
order, the linearized system%
\begin{equation*}
\frac{\partial f_{i}}{\partial t}+\mathbf{p}_{i}\cdot \nabla _{\mathbf{x}%
}f_{i}+\left( Lf\right) _{i}=0\text{ for }i\in \left\{ 1,...,N\right\}
\end{equation*}%
is obtained. Here $L$ is the linearized collision operator--$N\times N$
matrix--given by%
\begin{equation}
\left( Lf\right) _{i}=\sum\limits_{j,k,l=1}^{N}\frac{\Gamma _{ij}^{kl}}{%
R_{i}^{1/2}}%
(P_{ij}^{kl}f_{i}+P_{ji}^{kl}f_{j}-P_{kl}^{ij}f_{k}-P_{lk}^{ij}f_{l})\text{
for }i\in \left\{ 1,...,N\right\} \text{.}  \label{l22c}
\end{equation}%
Note that, in agreement with \cite{Be-17,Be-23g}, for bosons ($\alpha =0$)
and fermions ($\alpha =1$)%
\begin{equation*}
R=P(1+P)\text{ and }R=P(1-P)\text{,}
\end{equation*}%
respectively, while for semions ($\alpha =1/2$) 
\begin{equation*}
R=P\left( 1-\frac{P^{2}}{4}\right) \text{.}
\end{equation*}

\subsection{Some properties of the linearized collision operator \label{S4.1}%
}

Denoting 
\begin{equation*}
\Pi _{ij}^{kl}\left( g\right) =g_{i}g_{j}\Psi _{\alpha }\left( g_{k}\right)
\Psi _{\alpha }\left( g_{l}\right) -g_{k}g_{l}\Psi _{\alpha }\left(
g_{i}\right) \Psi _{\alpha }\left( g_{j}\right) \text{,}
\end{equation*}%
it can be observed that 
\begin{eqnarray}
P_{ij}^{kl} &=&\left. \frac{\partial \Pi _{ij}^{kl}\left( P+R^{1/2}f\right) 
}{\partial f_{i}}\right\vert _{f=0}=R_{i}^{1/2}\left( P_{j}\Psi _{\alpha
}\left( P_{k}\right) \Psi _{\alpha }\left( P_{l}\right) -P_{k}P_{l}\Psi
_{\alpha }^{\prime }\left( P_{i}\right) \Psi _{\alpha }\left( P_{j}\right)
\right)  \notag \\
&=&\frac{P_{i}P_{j}\Psi _{\alpha }\left( P_{k}\right) \Psi _{\alpha }\left(
P_{l}\right) }{R_{i}^{1/2}}\frac{\Psi _{\alpha }\left( P_{i}\right) -\Psi
_{\alpha }^{\prime }\left( P_{i}\right) P_{i}}{P_{i}\Psi _{\alpha }\left(
P_{i}\right) }R_{i}=\frac{P_{i}P_{j}\Psi _{\alpha }\left( P_{k}\right) \Psi
_{\alpha }\left( P_{l}\right) }{R_{i}^{1/2}}\text{,}  \label{l22d}
\end{eqnarray}%
for any indices $\left\{ i,j,k,l\right\} \subset \left\{ 1,...,N\right\} $,
since, by relations $\left( \ref{l20}\right) $ and $\left( \ref{l20a}\right) 
$,%
\begin{equation}
\frac{\Psi _{\alpha }\left( P_{i}\right) -\Psi _{\alpha }^{\prime }\left(
P_{i}\right) P_{i}}{P_{i}\Psi _{\alpha }\left( P_{i}\right) }R_{i}=1\text{
for }i\in \left\{ 1,...,N\right\} \text{.}  \label{l24}
\end{equation}

By relations $\left( \ref{l3}\right) $ and $\left( \ref{c14b}\right) $, the
relation 
\begin{equation}
P_{i}P_{j}\Psi _{\alpha }\left( P_{k}\right) \Psi _{\alpha }\left(
P_{l}\right) =P_{k}P_{l}\Psi _{\alpha }\left( P_{i}\right) \Psi _{\alpha
}\left( P_{j}\right)  \label{l23}
\end{equation}%
is obtained for any indices $\left\{ i,j,k,l\right\} \subset \left\{
1,...,N\right\} $ such that $\Gamma _{ij}^{kl}\neq 0$. Hence, by relations $%
\left( \ref{l6}\right) $, $\left( \ref{l22c}\right) $, $\left( \ref{l22d}%
\right) $, and $\left( \ref{l23}\right) $, we have the following lemma for
the weak form of the linearized collision operator.

\begin{lemma}
\label{L2}For any functions $g=g(\mathbf{x},\mathbf{p},t)$ and $f=f(\mathbf{x%
},\mathbf{p},t)$ the weak form of the linearized collision operator can be
recast as 
\begin{multline}
\left\langle g,Lf\right\rangle =\frac{1}{4}\sum\limits_{i,j,k,l=1}^{N}\Gamma
_{ij}^{kl}P_{i}P_{j}\Psi _{\alpha }\left( P_{k}\right) \Psi _{\alpha }\left(
P_{l}\right) \left( \frac{f_{i}}{R_{i}^{1/2}}+\frac{f_{j}}{R_{j}^{1/2}}-%
\frac{f_{k}}{R_{k}^{1/2}}-\frac{f_{l}}{R_{l}^{1/2}}\right) \\
\times \left( \frac{g_{i}}{R_{i}^{1/2}}+\frac{g_{j}}{R_{j}^{1/2}}-\frac{g_{k}%
}{R_{k}^{1/2}}-\frac{g_{l}}{R_{l}^{1/2}}\right) \text{.}  \label{l25}
\end{multline}%
The following proposition follows directly.
\end{lemma}

\begin{proposition}
\label{P1}The matrix $L$ is symmetric and positive semi-definite, i.e.,%
\begin{equation*}
\left\langle g,Lf\right\rangle =\left\langle Lg,f\right\rangle \text{ and }%
\left\langle f,Lf\right\rangle \geq 0
\end{equation*}%
for all functions $g=g(\mathbf{x},\mathbf{p},t)$ and $f=f(\mathbf{x},\mathbf{%
p},t)$.
\end{proposition}

Furthermore, by relation $\left( \ref{l25}\right) $, $\left\langle
f,Lf\right\rangle =0$ if and only if%
\begin{equation}
\frac{f_{i}}{R_{i}^{1/2}}+\frac{f_{j}}{R_{j}^{1/2}}=\frac{f_{k}}{R_{k}^{1/2}}%
+\frac{f_{l}}{R_{l}^{1/2}}  \label{c20c}
\end{equation}%
for all indices $\left\{ i,j,k,l\right\} \subset \left\{ 1,...,N\right\} $
such that $\Gamma _{ij}^{kl}\neq 0$. Denoting $f=R^{1/2}\phi $ in equality $%
\left( \ref{c20c}\right) $, relation $\left( \ref{c9c}\right) $ is obtained.
Hence, since $L$ is semi-positive,%
\begin{equation*}
Lf=0\text{ if and only if }f=R^{1/2}\phi \text{,}
\end{equation*}%
where $\phi $\ is a collision invariant $\left( \ref{c9c}\right) $. The
following proposition follows.

\begin{proposition}
\label{P2}For normal models the kernel of the linearized operator $L$ is 
\begin{eqnarray}
\ker L &=&\mathrm{span}\left(
R^{1/2},R^{1/2}p^{1},...,R^{1/2}p^{d},R^{1/2}\left\vert \mathbf{p}%
\right\vert ^{2}\right) \text{, }  \notag \\
\text{with }R &=&P\left( 1-\alpha P\right) \left( 1+\left( 1-\alpha \right)
P\right) \text{.}  \label{c21}
\end{eqnarray}
\end{proposition}

\begin{remark}
\label{R2} \textit{Generalized collision operator.} More generally,
considering the collision operator $\left( \ref{q1}\right) $, corresponding
results in Lemma $\ref{L2}$, and Proposition $\ref{P1}$ and $\ref{P2}$ for
the linearized collision operator $L$ can be obtained analogously. Indeed,
the linearized operator $L$ is symmetric and positive semi-definite.
However, if at least one $a_{mn}$ such that $m\neq n$ is nonzero, then for
normal models the kernel $\left( \ref{c21}\right) $ in Proposition $\ref{P2}$
has to be replaced by 
\begin{equation*}
\ker L=\mathrm{span}\left( R^{1/2}p^{1},...,R^{1/2}p^{d},R^{1/2}\left\vert 
\mathbf{p}\right\vert ^{2}\right) \text{.}
\end{equation*}
\end{remark}

\begin{remark}
\label{R4} \textit{Applications to half-space problems.} The general results
obtained for planar stationary half-space problems \cite{BCN-86, CGS-88,
BGS-06, Be-23d} for the discrete linearized equations obtained in \cite%
{BB-03, Be-08, Be-17}--cf. the results in \cite{Be-23d} applied to discrete
models--yield also for the Boltzmann equation for anyons--also for the
general collision operator $\left( \ref{q1}\right) $--presented here.
Indeed, consider the planar stationary system $\left( \ref{c4a}\right) $%
--for the linearized collision operator, possibly also with an inhomogeneous
term, see \cite{BB-03, Be-08, Be-17}-- for $x>0$. Assume the components $%
F_{i}\left( 0\right) $ of the distribution function at $x=0$ for which $%
p_{i}^{1}$ is positive to be given--possibly linearly depending on the
components of $F\left( 0\right) $ for which $p_{i}^{1}$ is negative. Then
results concerning the number of\ conditions needed for existence and/or
uniqueness of solutions--based on the signature of the restriction of the
quadratic form $\left\langle \cdot ,B\cdot \right\rangle $ to the kernel of $%
L$--in \cite{BB-03, Be-08, Be-17} can be applied. We stress that the results
presented in \cite{BB-03, Be-08, Be-17,Be-23d} can be applied also for the
Cauchy problem in the spatially homogenous case.
\end{remark}

\begin{remark}
\textit{Extensions to mixtures and/or multiple internal energy states.} The
results can--also for the general collision operator $\left( \ref{q1}\right) 
$--be extended to mixtures--including mixtures of anyons with different
fractional statistics, i.e., with different $\alpha \in \left[ 0,1\right] $%
--as well as particles with multiple energy levels, applying approaches
presented in \cite{BV-16, Be-16b, Be-18, Be-17}. Indeed, the key feature is
that to each component $F_{i}$ of the distribution function $F$ there will
be assigned not only a momentum $\mathbf{p}_{i}$, but also a species $a_{i}$
with species-dependent $\varepsilon _{a_{i}}=\alpha _{i}$ for $\alpha
_{i}\in \left[ 0,1\right] $, and possibly also an internal energy $I_{i}$.
The sets of admissible momentums--and possibly internal energies--may vary
for different species. At a formal level this extension seems merely to be a
matter of notation. Known normal models for discrete velocity models of the
Boltzmann equation, see \cite{BV-16, Be-16b, Be-17, Be-18} and references
therein, can be made use of (at least) in case of the collision operator $%
\left( \ref{l2}\right) $.
\end{remark}

\section{Conclusions \label{S5}}

A general discrete model of Boltzmann equation for anyons--or, Haldane
statistics--has been reviewed. As limiting cases the Nordheim-Boltzmann
equation for bosons and fermions appear.

The equilibrium distributions were characterized through a transcendental
equation and analytically solved for bosons, fermions, and semions. Trend to
equilibrium in the spatially homogeneous--were a certain equilibrium
distribution is approached--as well as the planar stationary case has been
shown.

The linearized collision operator was shown to be a symmetric, non-negative
operator, and its null-space--of the same dimension as the vector space of
the collision invariants--was characterized. Applications to the Cauchy
problem for linearized spatially homogeneous equation, as well as the
linearized steady half-space problem in a slab-symmetry, were then indicated
based on corresponding results for general discrete velocity models of the
linearized Boltzmann equation \cite{BB-03, Be-08, Be-10, Be-15, Be-17}.

Generalizations to more general collision operators, mixtures, particles
with different internal energy states, as well as assumptions of other
collision invariants have also been indicated and briefly discussed.

%
%

\end{document}